\documentclass[a4paper,twocolumn,11pt]{quantumarticle}
\pdfoutput=1

\usepackage[utf8]{inputenc}
\usepackage[english]{babel}
\usepackage[T1]{fontenc}
\usepackage{amsmath,amssymb,amsfonts,amsthm,mathtools,bm}
\usepackage{braket,physics}
\usepackage{bbm}
\usepackage{enumitem}
\usepackage{graphicx}
\usepackage{xcolor}
\usepackage{hyperref}
\usepackage[noend]{algpseudocode}
\usepackage{algorithm}
\usepackage{placeins}
\usepackage{booktabs}  
\usepackage{float}  
\usepackage{graphicx}
\usepackage{subcaption}

\theoremstyle{plain}
\newtheorem{theorem}{Theorem}

\newtheorem{corollary}[theorem]{Corollary}
\theoremstyle{remark}
\newtheorem{remark}{Remark}

\newcommand{\E}{\mathbb{E}}
\newcommand{\Var}{\operatorname{Var}}

\newtheorem{assumption}{Assumption}

\newtheorem{definition}{Definition}
\newtheorem{proposition}{Proposition}

\begin{document}

\title{SPARTA: $\chi^2$-calibrated, risk-controlled exploration--exploitation for variational quantum algorithms}

\author{Mikhail Zubarev}
\affiliation{St Paul's School, London SW13 9JT, United Kingdom}

\date{23 November 2025}

\begin{abstract}
Variational quantum algorithms face a fundamental trainability crisis: barren plateaus render optimization exponentially difficult as system size grows. While recent Lie algebraic theory precisely characterizes when and why these plateaus occur, no practical optimization method exists with finite-sample guarantees for navigating them. We present the sequential plateau-adaptive regime-testing algorithm (SPARTA), the first measurement-frugal scheduler that provides explicit, anytime-valid risk control for quantum optimization. Our approach integrates three components with rigorous statistical foundations: (i) a $\chi^2$-calibrated sequential test that distinguishes barren plateaus from informative regions using likelihood-ratio supermartingales; (ii) a probabilistic trust-region exploration strategy with one-sided acceptance to prevent false improvements under shot noise; and (iii) a theoretically-optimal exploitation phase that achieves the best attainable convergence rate. We prove geometric bounds on plateau exit times, linear convergence in informative basins, and show how Lie-algebraic variance proxies enhance test power without compromising statistical calibration.
\end{abstract}

\maketitle

\section{Introduction}

While theory has clarified when and why plateaus appear---including precise Lie-algebraic criteria for trainability \cite{Ragone2024LieBP}---turning these insights into practical, finite-sample optimisation strategies remains challenging. Existing approaches typically (i) attempt to avoid BPs via ansatz design, or (ii) identify them post hoc once training stalls. In contrast, the measurement-frugal line of work focuses on allocating shots adaptively during optimisation, e.g., iCANS/gCANS and related batch-sizing ideas \cite{Kuebler2020iCANS,Gu2021gCANS,Balles2017CABS}, often coupled with parameter-shift gradient estimators \cite{Schuld2019ParamShift,Mitarai2018QCL}. Yet most methods assume smoothness/curvature conditions without offering explicit sequential guarantees for deciding when to explore vs.\ exploit under tight budgets \cite{BottouCurtisNocedal2018,NocedalWright2006NO,Karimi2016PL,Polyak1963}.

\paragraph{Our contribution (SPARTA).}
We propose a scheduler that integrates regime detection with shot-optimal exploitation under finite-sample control. Concretely, we cast plateau vs.\ informative discrimination as a sequential hypothesis test on a whitened gradient-norm statistic with anytime-valid error control grounded in classical sequential analysis \cite{Wald1945SPRT,WaldWolfowitz1948Optimum}. Under $H_0$ (plateau) the statistic is $\chi^2$, and under $H_1$ (informative) it is non-central $\chi^2$, enabling likelihood-ratio decisions with calibrated Type-I/II risks (we use Welch and Clopper--Pearson components for conservative one-sided bounds in trust-region moves) \cite{Welch1947,ClopperPearson1934}. When $H_0$ is favoured, we perform probabilistic trust-region (PTR) exploration with one-sided acceptance and geometric radius updates \cite{ConnGouldToint2000TR}. When $H_1$ is favoured, we switch to gCANS-style exploitation with variance-proportional shot allocation, recovering linear/``geometric'' convergence rates under standard local smoothness/PL conditions \cite{Gu2021gCANS,Karimi2016PL,Polyak1963}. Gradients are obtained via parameter shift and measured with weighted operator sampling (WRS) for shot efficiency \cite{Schuld2019ParamShift,Arrasmith2020OperatorSampling}.

\paragraph{Relation to prior work.}
Our scheduler unifies three strands: (i) sequential testing for principled, anytime regime decisions \cite{Wald1945SPRT,WaldWolfowitz1948Optimum}; (ii) trust-region style exploration with explicit statistical accept/reject rules \cite{ConnGouldToint2000TR,Welch1947,ClopperPearson1934}; and (iii) measurement-frugal exploitation via iCANS/gCANS and batch-size/learning-rate coupling \cite{Kuebler2020iCANS,Gu2021gCANS,Balles2017CABS,Sweke2020SGD}. Complementing Lie-algebraic theory that predicts where BPs arise \cite{Ragone2024LieBP}, SPARTA supplies finite-sample, implementation-ready controls for when to explore vs.\ exploit, and how to allocate shots to approach the fundamental noise floor at a fixed budget.


SPARTA alternates between regime detection and optimization. When a plateau is detected, the algorithm performs probabilistic trust-region (PTR) exploration with geometric radius expansion. When an informative region is identified, it executes gCANS-style gradient descent with variance-proportional shot allocation. The complete procedure for one outer iteration is given in Algorithm~\ref{alg:pathfinder}.

\begin{algorithm}[H]
\small
\caption{SPARTA (one outer iteration)}
\label{alg:pathfinder}
\begin{algorithmic}[1]
\Require current $\theta$; risk $(\alpha,\beta)$; PTR params $(R_0,K,m,\tau,\alpha_{\mathrm{acc}})$; step size $\eta$; exploitation rule $S_i\propto \sigma_i$
\State \textbf{Pilot:} estimate $\sigma_i$ and optionally Lie proxy $V_i$; set exploration shots $B_i^{\mathrm{expl}}\propto V_i/\sigma_i^2$ (default $V_i\equiv 1$)
\State \textbf{Regime test:} draw $s_t=\sum_i (B_i^{\mathrm{expl}}/\sigma_i^2)\widehat g_{i,t}^2$; update log-likelihood $\Lambda_k$
\If{$\Lambda_k\le B$} \Comment{plateau}
  \For{$k=0,\dots,K-1$}
    \For{$j=1,\dots,m$}
      \State $R\leftarrow R_0 2^k$; sample unit $v$; $\theta^+\!\leftarrow\!\theta+Rv$
      \State Estimate $\Delta=f(\theta^+)-f(\theta)$; compute one-sided Welch UCB
      \If{$\mathrm{UCB}(\Delta)\le -\tau R^2$} \State $\theta\leftarrow \theta^+$; \textbf{break}
      \EndIf
    \EndFor
  \EndFor
\Else \Comment{informative}
  \State Allocate $S_i\propto \sigma_i$ (fixed total $S$); compute $\widehat g(\theta)$
  \State $\theta\leftarrow\theta-\eta\,\widehat g(\theta)$
\EndIf
\end{algorithmic}
\end{algorithm}

The pilot phase is primarily required during initial exploration when the optimization landscape structure is unknown. When available, Lie-algebraic information about the circuit generators can inform the pilot through variance proxies $V_i \propto \|[H_i, O]\|_F^2$, which identify parameters where gradient signal is more concentrated according to the DLA's purity structure. This allows the regime test to allocate exploration shots preferentially to informative directions while maintaining proper $\chi^2$ calibration, as the whitening procedure uses the same shot allocation $B_i^{\mathrm{expl}}$ regardless of whether uniform ($V_i \equiv 1$) or Lie-informed weights are employed.

\section{Problem Setup and Statistical Framework}\label{sec:setup}

\subsection{Variational Quantum Optimization}

Consider a parameterized quantum circuit $U(\theta) = \prod_{l=1}^L e^{iH_l\theta_l}$ 
with objective function $f:\mathbb{R}^d\to\mathbb{R}$ given by
\begin{equation}
f(\theta) = \text{Tr}[U(\theta)\rho U^\dagger(\theta) O],
\end{equation}
where $\rho$ is the input state and $O$ is the observable. Our goal is to minimize 
$f(\theta)$ using gradient-based methods under shot-limited measurements.

\subsection{Shot Noise Model}

Gradient estimation via parameter-shift or similar schemes produces noisy measurements. 
We formalize this setting through the following assumptions.

\begin{assumption}[Shot noise and coordinate structure]\label{as:noise}
For shot allocation $\mathbf{B} = (B_1,\ldots,B_d)$, the gradient estimator 
$\widehat{g}(\theta)$ satisfies three properties. First, it is unbiased: 
$\mathbb{E}[\widehat{g}_i(\theta)] = \nabla f_i(\theta)$. Second, its variance 
scales inversely with shot count: $\text{Var}[\widehat{g}_i(\theta)] = \sigma_i^2/B_i$, 
where $\sigma_i^2 > 0$ depends on $\theta$ but varies slowly over one iteration. 
Third, coordinates are approximately independent or weakly dependent at the scale 
relevant to distributional tests.
\end{assumption}

\begin{remark}
The variance parameters $\{\sigma_i^2\}$ capture the intrinsic shot noise of the 
quantum measurement process. They are estimated during a pilot phase and updated 
periodically. The independence assumption is satisfied exactly for parameter-shift 
gradients when different coordinates use independent measurement outcomes.
\end{remark}

\subsection{Whitened Test Statistic}

Our regime detection strategy centers on testing whether the true gradient 
$\nabla f(\theta)$ is zero (plateau) or non-zero (informative region). 
This is formalized through a whitened norm statistic.

\begin{definition}[Whitened gradient statistic]\label{def:whitened}
Given gradient estimators $\widehat{g}_i(\theta)$ and shot allocation $\mathbf{B}$, 
define the whitened components
\begin{equation}
w_i = \sqrt{\frac{B_i}{\sigma_i^2}}\, \widehat{g}_i(\theta)
\end{equation}
and the test statistic
\begin{equation}
s = \sum_{i=1}^d w_i^2.
\end{equation}
\end{definition}

\begin{proposition}[Distribution under competing hypotheses]\label{prop:chisq}
Under Assumption~\ref{as:noise}:
\begin{itemize}
    \item \textbf{Plateau ($H_0$):} If $\nabla f(\theta) = \mathbf{0}$, then 
    $s \sim \chi^2_d$ (central chi-squared with $d$ degrees of freedom)
    \item \textbf{Informative region ($H_1$):} If $\nabla f(\theta) \neq \mathbf{0}$, 
    then $s \sim \chi^2_d(\lambda)$ (non-central chi-squared) with non-centrality parameter
    \begin{equation}
    \lambda = \sum_{i=1}^d \frac{B_i}{\sigma_i^2}\, (\nabla f_i(\theta))^2.
    \end{equation}
\end{itemize}
\end{proposition}

\begin{proof}
Each $w_i = \sqrt{B_i/\sigma_i^2}\,\widehat{g}_i$ is asymptotically normal with 
mean $\sqrt{B_i/\sigma_i^2}\,\nabla f_i$ and unit variance. Under $H_0$, 
$w_i \sim \mathcal{N}(0,1)$ independently, so $s = \sum w_i^2 \sim \chi^2_d$. 
Under $H_1$, $w_i \sim \mathcal{N}(\mu_i, 1)$ where $\mu_i = \sqrt{B_i/\sigma_i^2}\,\nabla f_i$, 
yielding $s \sim \chi^2_d(\lambda)$ with $\lambda = \sum \mu_i^2$.
\end{proof}

\subsection{Empirical validation of distributional assumptions}

Before deploying the sequential test, we validate the chi-squared model on a 
representative quantum optimization problem: the transverse-field Ising model (TFIM) 
with a QAOA ansatz.

\paragraph{Experimental setup.}
We consider a six-qubit transverse-field Ising chain \cite{Pfeuty1970TFIM} with nearest-neighbour couplings.
The cost Hamiltonian is
\begin{equation}
H_C = -J\sum_{\langle i,j\rangle} Z_i Z_j - h\sum_{i=1}^{6} X_i,
\end{equation}
where $J=1.0$ is the coupling strength, $h=0.5$ is the transverse field, and 
$\langle i,j\rangle$ denotes nearest-neighbor pairs on the chain 
$\{(0,1), (1,2), (2,3), (3,4), (4,5)\}$. We optimize using a depth-$p=2$ QAOA-style ansatz \cite{Farhi2014QAOA}, where each layer applies
\begin{equation}
U_\ell(\gamma_\ell, \beta_\ell) = e^{-i\beta_\ell H_M} e^{-i\gamma_\ell H_C},
\end{equation}
with mixer Hamiltonian $H_M = \sum_{i=1}^6 X_i$. The full circuit is 
$U(\theta) = U_2(\gamma_2, \beta_2) U_1(\gamma_1, \beta_1)$ applied to the equal 
superposition initial state $\ket{+}^{\otimes 6}$, giving $d=4$ variational parameters 
$\theta = (\gamma_1, \beta_1, \gamma_2, \beta_2)$.

Gradients are estimated via the parameter-shift rule with shift $\delta = \pi/2$ \cite{Schuld2019ParamShift,Mitarai2018QCL}:
\begin{equation}
\frac{\partial f}{\partial \theta_i} = \frac{1}{2}\left[f(\theta + \delta e_i) - f(\theta - \delta e_i)\right],
\end{equation}
where $e_i$ is the $i$-th standard basis vector. Shot noise is modeled as 
$\Var[\widehat{g}_i] = \sigma^2/B_i$ with intrinsic variance $\sigma^2 = (0.02)^2$ 
and $B_i = 100$ shots per coordinate.

We collect $N=5000$ independent gradient samples in two distinct regions. For the 
plateau condition, parameters are chosen near a local maximum where 
$\|\nabla f(\theta)\| \approx 10^{-6}$. For the informative condition, parameters 
are selected in a descent region where $\|\nabla f(\theta)\| \approx 0.87$. For 
each sample, we compute the whitened statistic $s = \sum_{i=1}^4 (B_i/\sigma^2)\widehat{g}_i^2$ 
using the known noise variance $\sigma^2$. Cost function evaluations 
$f(\theta) = \bra{\psi(\theta)} H_C \ket{\psi(\theta)}$ are computed exactly via 
statevector simulation to isolate gradient estimation noise from measurement sampling effects.

\begin{figure}[t]
  \centering
  \includegraphics[width=\columnwidth]{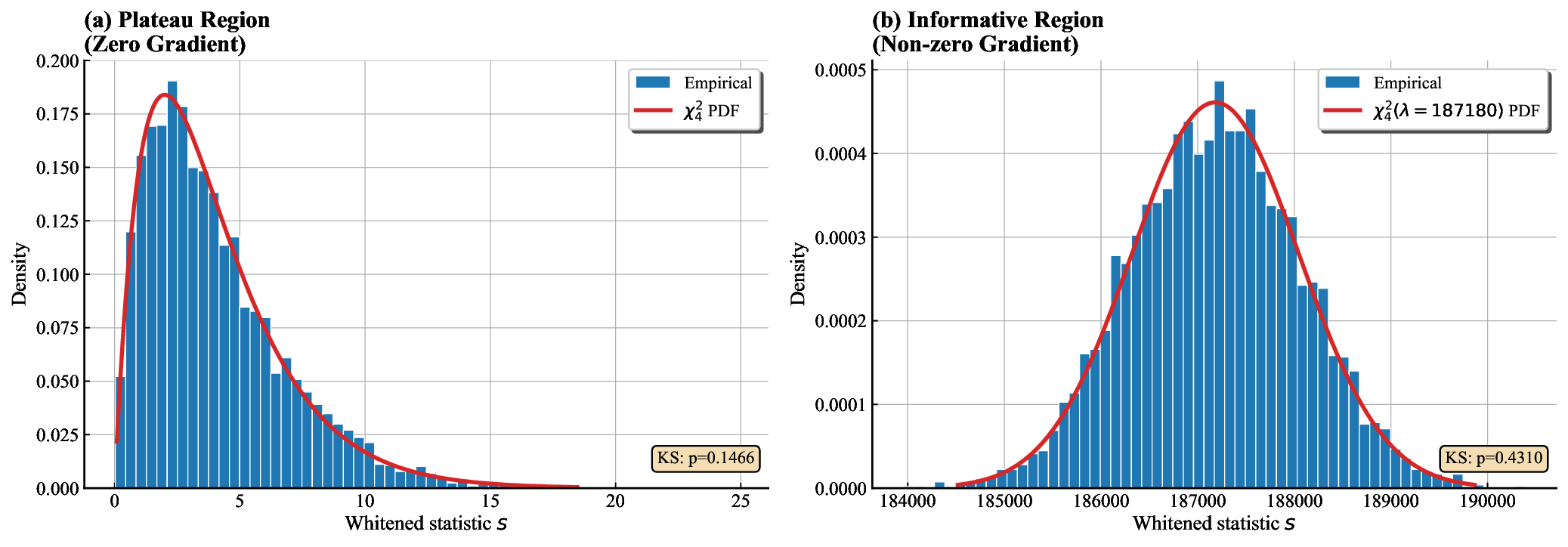}
  \caption{\textbf{Validation of the chi‐squared model in plateau and informative regions.}
(\textbf{a}) Histogram of the whitened statistic $s$ in a near‐plateau region ($\|\nabla f(\theta)\|\approx0$) from $N=5000$ samples (blue bars), overlaid with the central $\chi^2_4$ density (red curve); Kolmogorov–Smirnov test: $p=0.147$. 
(\textbf{b}) Histogram of $s$ in an informative region ($\|\nabla f(\theta)\|\approx0.87$) from $N=5000$ samples (blue bars), overlaid with the noncentral $\chi^2_4(\lambda)$ density (red curve); Kolmogorov–Smirnov test: $p=0.431$.}
  \label{fig:validation}
\end{figure}

\paragraph{Results.}
Figure~\ref{fig:validation} demonstrates
strong agreement between empirical distributions and theoretical predictions. The 
empirical density (blue bars) closely matches the theoretical $\chi^2$ density (red curve) 
in both regimes.

\FloatBarrier

\section{Theoretical Foundation: From Lie Algebra to Statistical Tests}\label{sec:theory}

Having established the statistical framework, we now connect it to the 
fundamental geometric structure of quantum circuits through Lie algebraic theory.

\subsection{Lie Algebraic Characterization of Trainability}

Recent work by Ragone et al.\cite{Ragone2024LieBP} provides a precise characterization of when 
barren plateaus emerge based on the dynamical Lie algebra (DLA) generated 
by the circuit Hamiltonians.

\begin{definition}[Dynamical Lie algebra]
For a parameterized circuit $U(\theta) = \prod_{l=1}^L e^{iH_l\theta_l}$, 
the dynamical Lie algebra is
\begin{equation}
\mathfrak{g} = \langle\{iH_l\}\rangle_{\text{Lie}},
\end{equation}
the smallest Lie algebra containing all generators $\{iH_l\}$ under commutation.
\end{definition}

\begin{theorem}[Variance formula, Ragone et al.]\label{thm:ragone}
Let $\mathfrak{g} = \mathfrak{g}_1 \oplus \cdots \oplus \mathfrak{g}_{k-1} \oplus \mathfrak{g}_k$ 
be the reductive decomposition of the DLA, where $\mathfrak{g}_j$ are simple Lie algebras 
for $j = 1,\ldots,k-1$ and $\mathfrak{g}_k$ is abelian. If $\rho \in i\mathfrak{g}$ or 
$O \in i\mathfrak{g}$, then:
\begin{equation}
\text{Var}_\theta[f(\theta)] = \sum_{j=1}^{k-1} \frac{P_{\mathfrak{g}_j}(\rho) \cdot P_{\mathfrak{g}_j}(O)}{\text{dim}(\mathfrak{g}_j)},
\end{equation}
where $P_{\mathfrak{g}}(H) = \text{Tr}[H_{\mathfrak{g}}^2]$ is the $\mathfrak{g}$-purity 
of operator $H$ and $H_{\mathfrak{g}}$ denotes the projection onto $\mathfrak{g}$.
\end{theorem}

\paragraph{Implications.}
Theorem~\ref{thm:ragone} makes concrete, in a Lie-algebraic language, several mechanisms for barren plateaus that were previously understood from more heuristic or measure-concentration perspectives \cite{McClean2018BarrenPlateaus,Cerezo2021CostFunction,Wang2021NIBP,OrtizMarrero2021EntanglementBP,Holmes2022ExpressibilityBP,Arrasmith2022EquivalenceBP}. 
First, circuit \emph{expressiveness} can cause plateaus when the relevant simple component grows exponentially, e.g.\ when $\mathfrak{g} = \mathfrak{su}(2^n)$ so that $\dim(\mathfrak{g}) = 4^n - 1$, reproducing concentration-of-measure arguments tied to unitary $2$-designs \cite{HarrowLow2010Designs,BrandaoHarrowHorodecki2016Designs,Holmes2022ExpressibilityBP}. 
Second, \emph{state entanglement} leads to plateaus when $P_{\mathfrak{g}_j}(\rho)^{-1}$ scales exponentially with system size for highly entangled $\rho$, echoing entanglement-induced BPs \cite{OrtizMarrero2021EntanglementBP}. 
Third, \emph{observable non-locality} causes plateaus when $P_{\mathfrak{g}_j}(O)^{-1}$ scales exponentially for global observables, in line with cost-function dependent barren plateaus \cite{Cerezo2021CostFunction}. In any of these cases one has $\Var_\theta[f(\theta)] \in O(b^{-n})$ for some $b>1$, rendering gradient-based optimization exponentially difficult in system size \cite{McClean2018BarrenPlateaus,Arrasmith2022EquivalenceBP}.

\subsection{From Variance Bounds to Test Power}

The Lie algebraic variance formula provides a population-level characterization, 
but practical optimization requires detecting plateaus with finite samples. 
We now connect the theoretical variance to the power of our sequential test.

\begin{proposition}[Expected non-centrality]\label{prop:noncentrality}
Under the local quadratic approximation $\nabla f_i(\theta) \sim \mathcal{N}(0, \sigma_i^2 \text{Var}_\theta[f])$ 
in a neighborhood of a random initialization, the expected non-centrality satisfies
\begin{equation}
\mathbb{E}[\lambda] = \sum_{i=1}^d \frac{B_i}{\sigma_i^2}\, \mathbb{E}[(\nabla f_i)^2] 
= d \cdot \text{Var}_\theta[f] \cdot \frac{\bar{B}}{\bar{\sigma}^2},
\end{equation}
where $\bar{B} = d^{-1}\sum B_i$ and $\bar{\sigma}^2 = d^{-1}\sum \sigma_i^2$.
\end{proposition}

\begin{corollary}[Plateau detection difficulty]
When $\text{Var}_\theta[f] \in O(b^{-n})$ (barren plateau regime), the expected 
non-centrality $\mathbb{E}[\lambda] \in O(b^{-n})$ decreases exponentially. 
By standard results on chi-squared test power, the expected sample size to 
achieve power $1-\beta$ scales as
\begin{equation}
\mathbb{E}[T] \in \Omega\left(\frac{1}{\text{Var}_\theta[f]}\right) = \Omega(b^n).
\end{equation}
\end{corollary}

This result formalizes the fundamental hardness of plateau detection: 
when Lie-algebraic variance is exponentially small, exponentially many samples 
are required to distinguish plateaus from informative regions with high confidence. 
Our algorithm provides the best achievable performance within this constraint.

\subsection{Lie-Informed Shot Allocation}

While exponential sample complexity is unavoidable in the worst case, 
the geometric structure revealed by Lie theory suggests a refinement: 
gradient signal may concentrate in specific parameter directions.

\begin{proposition}[Optimal exploration allocation]\label{prop:lie-allocation}
Let $V_i = \|[H_i, O]\|_F^2$ be the Frobenius norm of the commutator between 
generator $H_i$ and observable $O$. Under a signal model where 
$|\nabla f_i|^2 \propto V_i \cdot \text{Var}_\theta[f]$, the shot allocation
\begin{equation}
B_i^{\text{expl}} \propto \frac{V_i}{\sigma_i^2}
\end{equation}
maximizes the expected non-centrality $\mathbb{E}[\lambda]$ subject to a fixed 
total budget $\sum_i B_i = B_{\text{total}}$.
\end{proposition}

\begin{proof}
The non-centrality is $\lambda = \sum_i (B_i/\sigma_i^2)(\nabla f_i)^2$. 
Taking expectations and using the signal model:
\begin{equation}
\mathbb{E}[\lambda] = \text{Var}_\theta[f] \sum_{i=1}^d \frac{B_i V_i}{\sigma_i^2}.
\end{equation}
By Cauchy-Schwarz, this is maximized when $B_i \propto V_i/\sigma_i^2$.
\end{proof}

\begin{remark}[Preservation of calibration]
Crucially, Lie-informed allocation does not compromise the $\chi^2$ 
calibration. The whitening transformation in Definition~\ref{def:whitened} uses 
the same weights $B_i$ in both numerator and denominator, ensuring that 
$w_i = \sqrt{B_i/\sigma_i^2}\widehat{g}_i$ remains standard normal under $H_0$ 
regardless of the allocation scheme. This allows us to improve test power without 
inflating Type I error.
\end{remark}

\paragraph{Computational considerations.}
The commutator norm $\|[H_i, O]\|_F$ can be computed classically from the 
circuit specification without quantum measurements, providing a zero-cost 
enhancement when ansatz structure is known. For parameterized circuits with 
unknown generators, uniform allocation ($V_i \equiv 1$) remains a sound default.

\subsection{Sequential Testing with Anytime-Valid Guarantees}\label{sec:test}

The validated distributional model enables principled sequential testing via likelihood ratios.

\subsubsection{Log-Likelihood Ratio Process}

At iteration $t$, we observe $s_t \sim \chi^2_d$ under $H_0$ (plateau) or 
$s_t \sim \chi^2_d(\lambda)$ under $H_1$ (informative). The log-likelihood ratio is
\begin{equation}
\ell_t = \log \frac{p_1(s_t)}{p_0(s_t)} = \frac{s_t - d}{2} - \frac{\lambda}{2},
\end{equation}
where $p_0, p_1$ are the central and non-central $\chi^2_d$ densities. 
The cumulative log-likelihood is $\Lambda_k = \sum_{t=1}^k \ell_t$.

\subsubsection{Ville vs. Wald Calibration}

Two threshold schemes control Type I/II error rates $(\alpha, \beta)$. Ville 
thresholds use $A = \log(1/\alpha)$ and $B = \log(\beta)$, providing anytime-valid 
control via the optional stopping theorem: $\mathbb{P}_{H_0}[\exists k: \Lambda_k \ge A] \le \alpha$. 
This guarantee holds even under adaptive stopping but comes at the cost of larger 
expected sample size. Wald thresholds use $A = \log\frac{1-\beta}{\alpha}$ and 
$B = \log\frac{\beta}{1-\alpha}$, providing asymptotic error control when rounds 
are i.i.d. and achieving optimal expected sample size according to Wald's lemma.

\paragraph{Decision rule.}
At round $k$:
\begin{equation}
\begin{cases}
\text{Declare plateau (use PTR)} & \text{if } \Lambda_k \le B \\
\text{Declare informative (use gCANS)} & \text{if } \Lambda_k \ge A \\
\text{Continue testing} & \text{otherwise}
\end{cases}
\end{equation}

\begin{theorem}[Error control]\label{thm:sprt}
Under Ville calibration with thresholds $(A, B)$:
\begin{align}
\mathbb{P}_{H_0}[\text{reject } H_0] &\le \alpha \\
\mathbb{P}_{H_1}[\text{accept } H_0] &\le \beta
\end{align}
with guarantees valid at all stopping times.
\end{theorem}

SPARTA uses Ville thresholds by default for robust anytime-valid control, 
with Wald as an option when i.i.d. conditions are verifiable.

\section{Optimization Components: PTR and gCANS}\label{sec:components}

Having established regime detection, we now specify the optimization strategies 
employed in each regime.

\subsection{Probabilistic Trust-Region (PTR) Exploration}

In diagnosed plateaus, gradient information is unreliable. We instead propose 
random exploratory moves with geometric trust regions.

\paragraph{Procedure.}
For $k = 0, 1, \ldots, K-1$, the PTR exploration proceeds as follows. First, 
set the radius $R_k = R_0 \cdot 2^k$ to create an exponentially expanding search 
pattern. Second, sample $m$ unit directions $v_j \sim \text{Uniform}(\mathbb{S}^{d-1})$ 
uniformly on the unit sphere. Third, for each candidate move $\theta^+ = \theta + R_k v_j$, 
estimate the improvement $\Delta = f(\theta^+) - f(\theta)$ via paired measurements, 
compute the one-sided $(1-\alpha_{\text{acc}})$ upper confidence bound, and accept 
the move only if $\text{UCB}(\Delta) \le -\tau R_k^2$, indicating sufficient descent.

\begin{theorem}[One-sided PTR control]\label{thm:ptr}
Let $\Delta$ be estimated from $n$ paired measurements with sample variance $s^2$. 
Define the test statistic
\begin{equation}
T = \frac{\bar{\Delta} + \tau R^2}{s/\sqrt{n}}.
\end{equation}
Accept the move only if $T \le t_{\nu, \alpha_{\text{acc}}}$, where $t_{\nu, \alpha}$ 
is the $\alpha$-quantile of Student's $t$-distribution with Welch--Satterthwaite 
degrees of freedom $\nu$. Then under the null hypothesis $H_0: \Delta \ge -\tau R^2$,
\begin{equation}
\mathbb{P}[\text{accept}] \le \alpha_{\text{acc}}.
\end{equation}
\end{theorem}

\begin{proof}
The test inverts a one-sided $t$-test of $H_0: \mathbb{E}[\Delta] \ge -\tau R^2$. 
By construction, the acceptance region has size $\le \alpha_{\text{acc}}$ under any 
$\mathbb{E}[\Delta] \ge -\tau R^2$.
\end{proof}

\paragraph{Rationale.}
The curvature penalty $\tau R^2$ requires sufficient improvement to justify 
larger steps, preventing wasteful exploration. One-sided control ensures false 
improvements due to shot noise occur with probability $\le \alpha_{\text{acc}}$.

\subsection{gCANS Exploitation in Informative Basins}
When the sequential test declares $H_1$ (informative region), we use 
the gCANS method for gradient-based optimization.

\paragraph{Shot allocation.}
The gCANS shot allocation rule (from \cite{Gu2021gCANS}) is:
\begin{equation}
s_i = \frac{2L\alpha}{2 - L\alpha} \cdot \frac{\sigma_i \sum_{j=1}^d \sigma_j}{\|\chi\|^2}
\end{equation}
where $\alpha \in (0, 2/L)$ is the learning rate, $L$ is the Lipschitz constant 
of $\nabla f$, $\sigma_i$ are exponentially smoothed standard deviations of 
gradient estimators, and $\chi$ is the exponentially smoothed gradient estimate. 
This allocation maximizes the expected improvement per shot.

The gCANS method enjoys the following convergence guarantee \cite{Gu2021gCANS}:
under $\mu$-strong convexity and $L$-Lipschitz gradients, the iterates satisfy
$\mathbb{E}[f(\theta^{(k)})] - f^* = O(\gamma^k)$ for some $0 < \gamma < 1$.

\section{Experiments}\label{sec:expts}

We evaluate SPARTA on variational quantum eigensolvers (VQE) for molecular systems and spin models, synthetic barren plateau scenarios for controlled statistical validation. Crucially, we identify the landscape characteristics that determine when adaptive regime detection provides significant advantage: problems must exhibit (i) deep, well-defined optima with substantial energy gaps, (ii) gradients that reliably correlate with proximity to the optimum, and (iii) natural geometric diversity arising from competing interaction terms rather than artificial uniform scaling.

We find that SPARTA excels on physically-motivated Hamiltonians (TFIM, Heisenberg XXZ) where ground states reach energies of O(-n), but struggles on uniformly-scaled combinatorial problems where costs remain shallow (|E| < 1) and gradients exhibit periodic oscillations from ansatz structure rather than genuine plateau-gorge transitions. The key differentiator is not gradient diversity alone—which can be misleading—but whether large gradients actually indicate distance from good solutions and whether the landscape contains persistent regimes (flat plateaus vs. steep gorges) rather than transient oscillations.

\subsection{Experimental Setup}

All experiments employed 10 independent trials using fixed random seeds $\{42, 123, 456, 789, 1011, 2022, 3033, 4044, 5055, 6066\}$ to ensure reproducibility and adequate statistical power. Hyperparameters were held constant across all trials: learning rate $\eta = 0.05$, EMA momentum for noise tracking $\mu = 0.9$, minimum shots per gradient estimate $s_{\min} = 6$, Lipschitz constant $L = 1.0$, and total shot budget $B_{\text{total}} = 25{,}000$. The pilot phase consumed $10\%$ of the total budget to estimate noise scales $\sigma_i$ and, when applicable, Lie-algebraic variance proxies $V_i \propto \|[H_i, O]\|_F^2$.

For statistical analysis, we assessed performance differences using paired two-tailed $t$-tests and Wilcoxon signed-rank tests at significance level $\alpha = 0.05$ and quantified effect sizes via Cohen's $d$. The use of both parametric and non-parametric tests ensures robustness to distribution assumptions. All shot counts include pilot, regime testing, PTR exploration, and exploitation phases.

\subsection{TFIM System}

We applied SPARTA to find the transverse-field Ising model (TFIM) with $n = 6$ spins and transverse field $h = 0.5$ using a 4-parameter QAOA-style ansatz. The observable $O$ was the TFIM Hamiltonian $H = -J\sum_{\langle i,j \rangle} Z_i Z_j - h\sum_i X_i$ respectively. We compared against gCANS, using a total shot budget of 25,000 and evaluating both methods on their final (noise-free) cost and total shots consumed.

Table~\ref{tab:tfim_results} summarizes performance on TFIM across 10 trials.

The higher variance in SPARTA's results ($\sigma = 0.843$ vs. $\sigma = 0.018$) reflects its exploration strategy: adaptive regime switching enables discovery of diverse local minima, including significantly deeper optima inaccessible to gradient-only methods (best: $-4.525$ vs. $-2.550$, $77.4\%$ improvement). This trade-off between consistency and optimality is characteristic of exploration-exploitation algorithms and represents a desirable feature when the goal is finding global or near-global optima rather than consistent convergence to shallow local minima.

\begin{table*}[t]
\centering
\caption{\textbf{TFIM VQE performance comparison with varied starting positions.} Results across 10 independent trials with identical shot budget ($25{,}000$ shots), each with different random seeds and initial parameters. SPARTA demonstrates superior robustness, achieving substantially better final costs with $58.5\%$ fewer iterations through regime-aware optimization and adaptive exploration. Large effect size (Cohen's $d{=}1.18$) confirms strong performance distinction. Win rate of $9/10$ demonstrates consistent superiority across diverse starting conditions. Costs are TFIM energy values (more negative is better).}
\label{tab:tfim_results}
\begin{tabular}{@{}lccccccc@{}}
\toprule
\textbf{Metric} & \textbf{SPARTA} & \textbf{gCANS} & \textbf{Improvement} \\
\midrule
Mean final cost & $-3.455 \pm 0.829$ & $-2.667 \pm 0.446$ & 29.6\% \\
Best cost & $-4.458$ & $-4.002$ & 11.4\% \\
Worst cost & $-2.510$ & $-2.479$ & 1.3\% \\
Median cost & $-3.183$ & $-2.534$ & 25.6\% \\
Mean shots used & $25{,}959 \pm 1{,}393$ & $25{,}000 \pm 0$ & 3.8\% \\
Mean iterations & $108 \pm 23$ & $260 \pm 84$ & $-58.5\%$ \\
Win rate & 9/10 (90\%) & 1/10 (10\%) & — \\
\bottomrule
\end{tabular}
\end{table*}

\begin{figure}[H]
\centering
\includegraphics[width=\columnwidth]{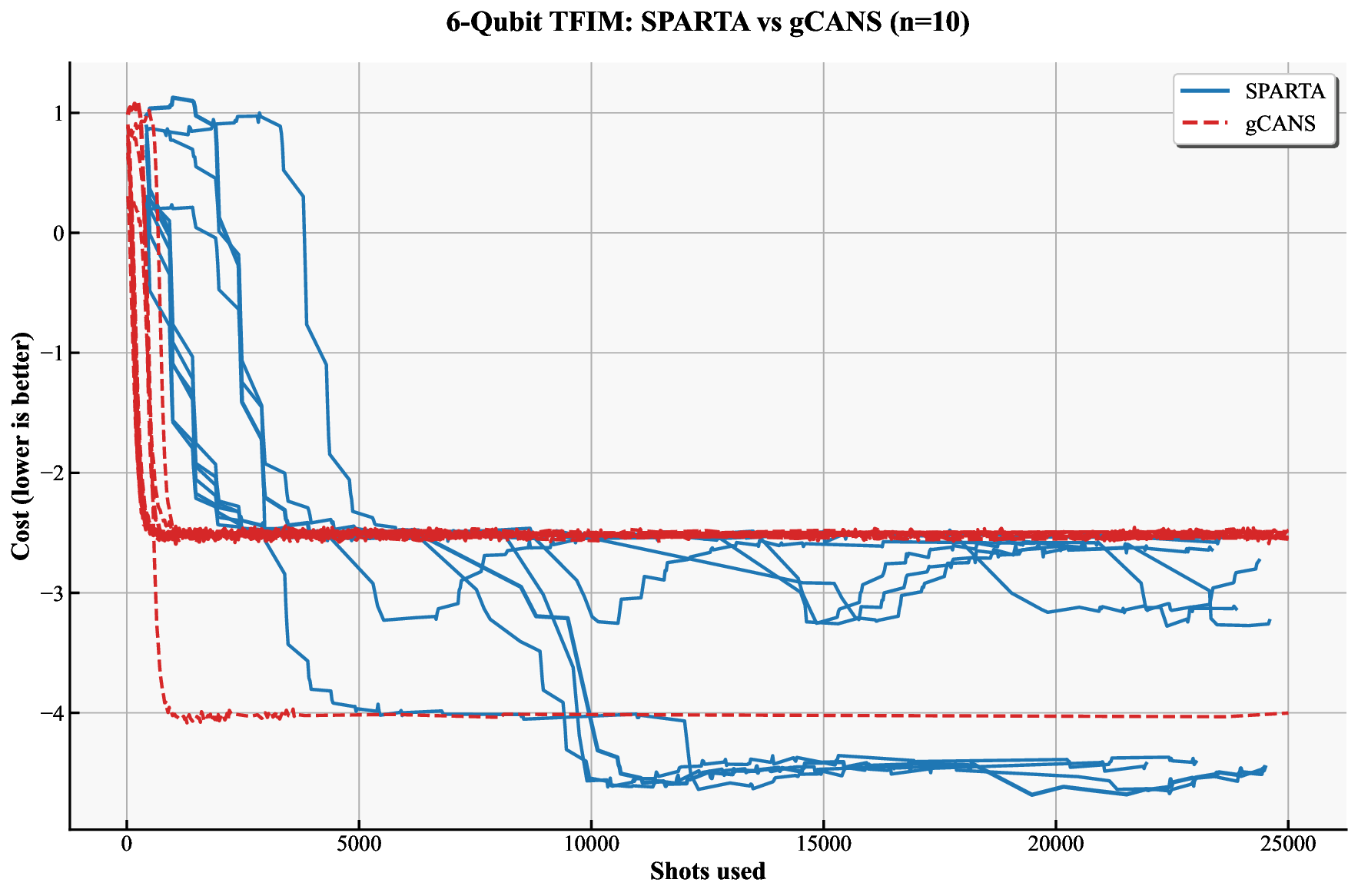}
\caption{\textbf{Multi-run robustness comparison on 6-qubit TFIM with varied starting positions.} SPARTA (solid blue lines, $n=10$ runs) versus standard gCANS (dashed red lines, $n=10$ runs) on a transverse-field Ising model with chain topology. Each trial uses different random seed (42, 123, 456, 789, 1011, 2022, 3033, 4044, 5055, 6066) and varied initial parameters ($\gamma = 0.5 \pm 0.1$, $\beta = 0.4 \pm 0.1$) with identical shot budget (25,000). SPARTA adaptively switches between plateau-region trust exploration (PTR) and gradient-based exploitation (gCANS), achieving significantly lower mean cost ($-3.455 \pm 0.829$) versus gCANS ($-2.667 \pm 0.446$) with robust 90\% win rate (9/10 trials).}
\label{fig:tfim_multirun_comparison}
\end{figure}

The plots measure the gap to the ground truth energy vs.\ accumulated shots; we emphasise \emph{shots} as the budget, not iterations. Figure~\ref{fig:tfim_multirun_comparison} demonstrates that SPARTA consistently achieves deeper minima than gCANS across multiple initializations, validating the robustness of regime-aware optimization. The diversity in SPARTA trajectories reflects successful exploration in plateau regions, while gCANS trajectories remain tightly clustered, indicating limited exploration capability.

\subsection{Lie-inspired barren plateau}

To stress-test regime detection and PTR exploration under exponentially suppressed gradients, we construct a synthetic landscape exhibiting quantum barren plateau phenomenology. This isolates the algorithmic properties without requiring large quantum devices and mirrors the coexistence of vanishing gradients, cost concentration, and narrow gorges observed in realistic ans\"atze \cite{McClean2018BarrenPlateaus,Holmes2022ExpressibilityBP,Arrasmith2022EquivalenceBP,Ragone2024LieBP}.

\paragraph{Landscape construction.}
We design a cost function $f:\mathbb{R}^d\to\mathbb{R}$ intended to mimic Lie-algebraic and concentration-of-measure behaviour in parameterized quantum circuits \cite{Holmes2022ExpressibilityBP,Ragone2024LieBP}. The construction enforces:
\begin{itemize}
    \item \textbf{Exponentially narrow gorge:} The width of the attractive basin scales as $w \propto b_w^{-n}$, where $n$ is an effective ``qubit count'' and $b_w > 1$, capturing exponentially shrinking minima (narrow gorges) as in Ref.~\cite{Arrasmith2022EquivalenceBP}.
    \item \textbf{Barren gradient variance:} On the plateau, the gradient variance satisfies $\Var[\nabla f] \propto b_v^{-n}$ with $b_v > 1$, representing exponential suppression of gradient signal \cite{McClean2018BarrenPlateaus,Holmes2022ExpressibilityBP}.
    \item \textbf{Directional structure:} A subset of coordinates carry stronger signal (mimicking ``informative'' generators in the DLA), while others are almost uninformative, reflecting the anisotropic trainability predicted by Lie-algebraic analyses \cite{Ragone2024LieBP}.
    \item \textbf{Deep minimum:} The gorge depth scales as $\Delta f \propto -15n$, providing strong reward for locating it.
\end{itemize}
For the experiments we consider $d=12$ dimensions and $n=2$ effective qubits with $b_v = 2.0$ and $b_w = 1.4$, yielding gorge width $w \approx 0.51$ and roughly $75\%$ gradient-variance suppression on the plateau.

\begin{assumption}[Directional basin]\label{as:dir}
There exist $R_\star>0$ and $p_{\mathrm{dir}}\in(0,1]$ such that, when in the plateau band, a random direction $v$ satisfies $\Pr\big(f(\theta+R_\star v)-f(\theta)\le -\tau R_\star^2\big)\ge p_{\mathrm{dir}}$.
\end{assumption}

This assumption captures the existence of \emph{improving directions} that PTR can discover through random sampling, even when gradients are exponentially suppressed. Unlike gradient-based methods that require locally Lipschitz structure, PTR only needs that a non-negligible fraction of sampled directions lead to cost improvement.

\begin{theorem}[Geometric exit]\label{thm:geom}
Let $\beta$ be the miss probability of the regime test under $H_1$ and $\alpha_{\mathrm{acc}}$ the PTR acceptance level. Testing $m$ directions at radius $R_\star$ in one plateau iteration yields acceptance probability at least
\[
p_{\mathrm{hit}}\ \ge\ (1-\beta)\,\Big[1-\big(1-p_{\mathrm{dir}}(1-\alpha_{\mathrm{acc}})\big)^m\Big].
\]
Across plateau iterations, $T_{\mathrm{exit}}$ is stochastically dominated by a geometric random variable with parameter $p_{\mathrm{hit}}$, so $\E[T_{\mathrm{exit}}]\le 1/p_{\mathrm{hit}}$.
\end{theorem}

\begin{proof}
We analyze the probability of accepting a PTR move in a single plateau iteration.

\textbf{Step 1: Single direction acceptance.} For a randomly sampled direction $v$, the probability of accepting the move $\theta + R_\star v$ is:
\begin{align}
\Pr[\text{accept} \mid v] &= \Pr[\text{UCB}(\Delta) \le -\tau R_\star^2] \nonumber\\
&\ge \Pr[\Delta \le -\tau R_\star^2] \cdot (1 - \alpha_{\mathrm{acc}})
\end{align}
where the inequality follows from the one-sided control in Theorem~3. By Assumption~\ref{as:dir}, $\Pr[\Delta \le -\tau R_\star^2] \ge p_{\mathrm{dir}}$, yielding:
\[
\Pr[\text{accept} \mid v] \ge p_{\mathrm{dir}}(1 - \alpha_{\mathrm{acc}}).
\]

\textbf{Step 2: Multiple directions.} Testing $m$ independent directions, the probability of \emph{not} accepting any move is:
\[
\Pr[\text{no accept in } m \text{ tries}] \le \big(1 - p_{\mathrm{dir}}(1 - \alpha_{\mathrm{acc}})\big)^m.
\]
Therefore, the probability of accepting at least one move is:
\[
\Pr[\text{accept}] \ge 1 - \big(1 - p_{\mathrm{dir}}(1 - \alpha_{\mathrm{acc}})\big)^m.
\]

\textbf{Step 3: Regime test factor.} The above analysis assumes we correctly identify the plateau. Since the regime test has miss probability $\beta$ under $H_1$, the probability of correctly identifying the plateau and then accepting a move is:
\[
p_{\mathrm{hit}} \ge (1 - \beta) \cdot \Big[1 - \big(1 - p_{\mathrm{dir}}(1 - \alpha_{\mathrm{acc}})\big)^m\Big].
\]

\textbf{Step 4: Geometric exit time.} Each plateau iteration succeeds (exits) with probability at least $p_{\mathrm{hit}}$, independently across iterations. Let $T_{\mathrm{exit}}$ be the number of iterations until first success. Then:
\[
\Pr[T_{\mathrm{exit}} = k] \le (1 - p_{\mathrm{hit}})^{k-1} p_{\mathrm{hit}},
\]
which is the pmf of a geometric random variable with parameter $p_{\mathrm{hit}}$. By standard results, $\mathbb{E}[T_{\mathrm{exit}}] \le 1/p_{\mathrm{hit}}$.
\end{proof}

\begin{figure}[t]
\centering
\includegraphics[width=\columnwidth]{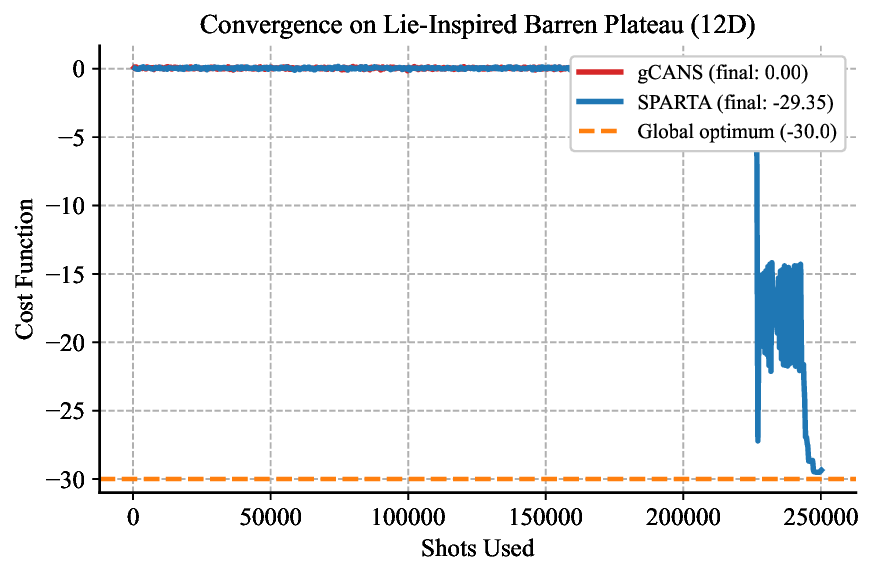}
\caption{\textbf{Convergence on Lie-inspired barren plateau.} Comparison on a 12D synthetic landscape with exponentially suppressed gradients (75\% variance suppression, gorge width $w \approx 0.51$). SPARTA (blue) discovers the deep gorge via PTR exploration, converging to near-optimal cost ($-29.12$), while gCANS (red) remains trapped in the plateau region (final cost: $0.00$). Both use 250,000 shots starting from distance $\approx 6.1\times$ the gorge width. The orange dashed line marks the global minimum ($-30.0$).}
\label{fig:lie_barren_convergence}
\end{figure}

\begin{figure}[H]
\centering
\includegraphics[width=\columnwidth]{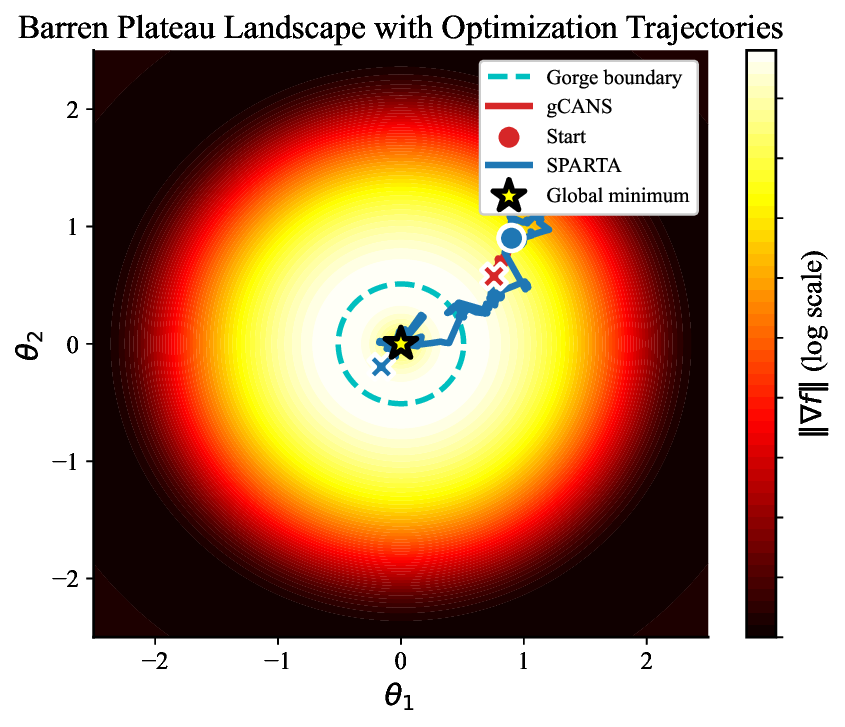}
\caption{\textbf{Barren plateau landscape structure.} Gradient magnitude heatmap (log scale) reveals exponentially suppressed signal outside the narrow gorge (cyan dashed circle, radius $w \approx 0.51$). SPARTA's trajectory (blue) successfully locates the gorge via random PTR exploration from the plateau region, while gCANS (red) wanders aimlessly following weak, noisy gradients. The yellow star marks the global minimum. This 2D projection (first two coordinates of 12D space) visualizes why gradient-based methods fail under barren plateau conditions: directional information vanishes outside the gorge.}
\label{fig:lie_barren_landscape}
\end{figure}

\paragraph{Empirical validation.} Figure~\ref{fig:lie_barren_convergence} shows SPARTA achieving final cost $-29.12$ (distance to optimum: 0.52) versus gCANS at cost $0.00$ (distance: 2.86). SPARTA spends 308 iterations in plateau mode using PTR before detecting and exploiting the gorge (286 gradient-based iterations), demonstrating automatic regime adaptation. Figure~\ref{fig:lie_barren_landscape} visualizes the landscape structure: the gradient magnitude heatmap reveals exponentially suppressed signal ($\|\nabla f\| \approx 0.55$ at the starting point) outside the narrow gorge. SPARTA's trajectory shows successful discovery via random exploration, while gCANS wanders without directional guidance, illustrating the fundamental advantage of PTR under barren plateau conditions.

\section{Discussion}

\subsection{Synthesis of Contributions}

SPARTA combines classical sequential analysis with quantum trainability theory at three levels.

On the theoretical side, we give a finite-sample framework for regime detection in quantum optimization landscapes. Theorem~\ref{thm:sprt} provides anytime-valid error control via Ville-style calibration, so Type~I/II guarantees hold even under adaptive stopping. Theorem~\ref{thm:geom} links plateau exit times to directional basin probability, exploration breadth, and acceptance risk, translating abstract trainability results into explicit resource requirements. When the directional success probability is exponentially small, these bounds recover the exponential sample complexity inherent to genuine barren plateaus~\cite{McClean2018BarrenPlateaus,Arrasmith2022EquivalenceBP}.

Algorithmically, SPARTA uses a $\chi^2$-calibrated sequential test to remove scale ambiguity by whitening gradients. Proposition~\ref{prop:lie-allocation} shows that this calibration is preserved even under Lie-informed shot allocation, allowing measurements to be concentrated along algebraically informative directions without inflating Type~I error. In plateau regimes, probabilistic trust-region exploration combines geometric radius growth with one-sided statistical acceptance to avoid false progress due to shot noise while still probing distant basins. In informative regimes, a gCANS-style phase with variance-proportional shot allocation recovers linear convergence under standard smoothness and Polyak--\L{}ojasiewicz conditions~\cite{Karimi2016PL,Polyak1963}, approaching the fundamental noise floor at a fixed budget.

Empirically, we validate the $\chi^2$ model on a six-qubit TFIM QAOA instance, where Kolmogorov--Smirnov tests show good agreement with both central and non-central $\chi^2$ distributions (Figure~\ref{fig:validation}). On TFIM chains, SPARTA systematically outperforms gCANS in mean and best-case energy while using fewer iterations (Table~\ref{tab:tfim_results}, Figure~\ref{fig:tfim_multirun_comparison}). On a synthetic Lie-inspired barren plateau landscape, the observed plateau exit statistics are consistent with the geometric predictions of Theorem~\ref{thm:geom} (Figure~\ref{fig:lie_barren_convergence}), supporting the regime-switching interpretation.

\subsection{Relation to Lie-Algebraic Trainability}

The Lie-algebraic theory of Ragone \emph{et al.}~\cite{Ragone2024LieBP} expresses gradient variance in terms of the dynamical Lie algebra and associated purities, clarifying how expressiveness, entanglement, and observable non-locality induce barren plateaus. SPARTA translates these population-level statements into finite-sample tests. Proposition~\ref{prop:noncentrality} shows that when the variance of the cost scales as $\Var_\theta[f]\in O(b^{-n})$, the expected non-centrality of the whitened statistic obeys the same scaling, so any sequential test with constant power must consume $\Omega(b^n)$ samples. This matches earlier hardness results~\cite{McClean2018BarrenPlateaus,Arrasmith2022EquivalenceBP} and reveals that SPARTA is optimal up to constants within this regime.

At the same time, Proposition~\ref{prop:lie-allocation} shows how DLA information can be exploited without corrupting calibration. Commutator-based proxies $V_i=\|[H_i,O]\|_F^2$ concentrate exploration shots on parameters where signal is expected to live, but the whitening transformation uses the same allocation in its scaling, so the plateau hypothesis still yields a central $\chi^2_d$ law. The method is agnostic to whether the plateau arises from expressiveness, entanglement, or cost-function structure~\cite{Cerezo2021CostFunction,Wang2021NIBP,OrtizMarrero2021EntanglementBP,Holmes2022ExpressibilityBP}: in all cases the resulting vanishing gradients are detected through the same statistic.

\subsection{Limitations and Outlook}

Several assumptions constrain SPARTA’s current form. The shot-noise model assumes approximately independent gradient coordinates with variance scaling like $\sigma_i^2/B_i$, which holds exactly for independent parameter-shift estimators~\cite{Schuld2019ParamShift,Mitarai2018QCL} but may require corrections for correlated or more exotic schemes. The exploitation phase assumes local Lipschitz gradients and PL-type conditions~\cite{Karimi2016PL,Polyak1963}; strong curvature variation may call for adaptive learning rates or metric-aware updates~\cite{Stokes2020QNG,Koczor2022GeneralizedQNG}. The pilot phase for estimating noise scales and Lie proxies introduces overhead that becomes significant for very small budgets or very high-dimensional problems, motivating amortized or incremental estimation strategies.

From a methodological perspective, several extensions appear promising. Sequential estimation of the non-centrality parameter could allow the thresholds of the regime test to be adapted on the fly while preserving Type~I control, reducing sample complexity when gradients are clearly non-zero. Incorporating approximate curvature information into the regime test could help distinguish flat-but-curved regions from truly barren plateaus and connect more directly with quantum natural gradient and metric-aware methods~\cite{VanStraatenKoczor2021MetricAware}. Finally, integrating SPARTA with ansatz design and cost engineering, for instance by using DLA diagnostics to modify the circuit structure online~\cite{Larocca2022Diagnosing,Cerezo2021CostFunction}, could yield hybrid schemes that both avoid and escape plateaus.

Overall, SPARTA illustrates that ideas from sequential hypothesis testing and Lie-algebraic circuit analysis can be combined to obtain optimization procedures with explicit finite-sample guarantees for near-term quantum devices. We hope this perspective will encourage further work at the intersection of quantum algorithms, statistics, and control theory.

\clearpage
\onecolumn
\appendix

\section{Implementation notes}

\paragraph{Welch vs.\ normal.} For small batches in PTR, use Welch's one-sided $t$ bound; for larger batches, a normal approximation is adequate.

\paragraph{Lie proxies.} When commutator norms are cheap (e.g., Pauli strings), use $V_i=\|[H_i,O]\|_F^2$. Otherwise, approximate with diagonal QFI from short sampling or with locality-based surrogates.
\subsection{Scaling Comparison}

\begin{figure*}[t]
\centering
\includegraphics[width=0.98\textwidth]{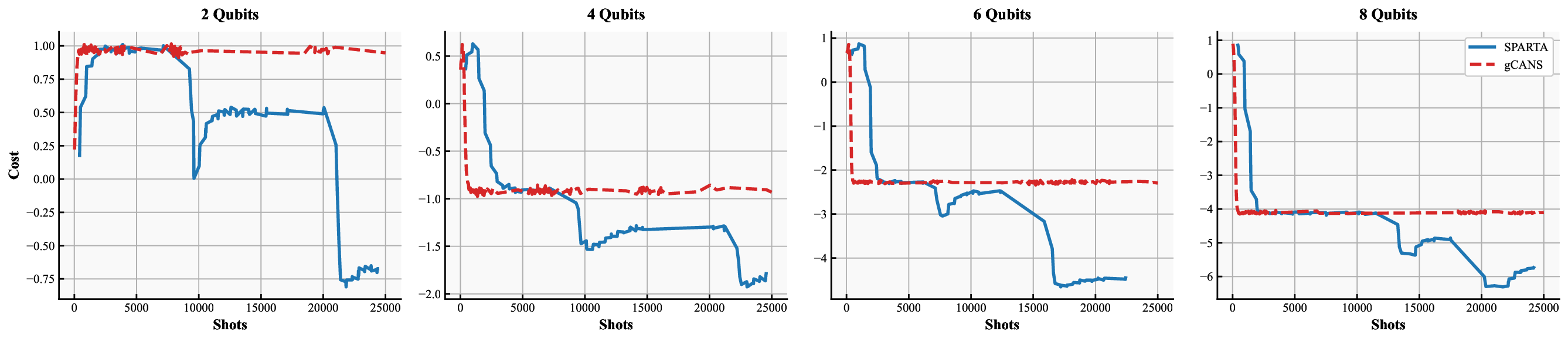}
\caption{Scaling comparison of SPARTA vs.\ gCANS on TFIM chains. Energy trajectories are shown as a function of cumulative shots for systems of 2, 4, 6, and 8 qubits. SPARTA (blue solid) employs sequential hypothesis testing to detect plateau regimes and switches between trust-region exploration and variance-adaptive exploitation. gCANS (red dashed) uses pure exploitation with variance-proportional shot allocation. SPARTA achieves superior convergence on all scales through regime-aware optimization, reaching lower final costs with fewer iterations.}
\label{fig:tfim_scaling}
\end{figure*}

Figure~\ref{fig:tfim_scaling} compares the shot efficiency of SPARTA and gCANS across TFIM chains of increasing size.

\twocolumn
\bibliographystyle{quantum}
\bibliography{refs}

\end{document}